\documentclass[11pt]{llncs}

\usepackage{fullpage}
\usepackage{epsfig}
\usepackage{graphics}
\usepackage{amsmath,setspace}
\usepackage{booktabs}
\usepackage{array}
\usepackage{multirow}
\usepackage{amssymb}
\usepackage{caption}
\usepackage{xspace}

\usepackage[dvipsnames,usenames]{color}
\usepackage[colorlinks=true,urlcolor=Blue,citecolor=Green,linkcolor=BrickRed]{hyperref}
\newcommand{\kSAT}{\textsc{$k$-SAT}\xspace}
\newcommand{\kSUM}{\textsc{$k$-Sum}\xspace}
\newcommand{\SSUM}{\textsc{Subset Sum}\xspace}
\newcommand{\bipath}{\textsc{Bicriteria $s,t$-Path}\xspace}
\newcommand{\EkPath}{\textsc{Exact $k$-Path}\xspace}
\newcommand{\EkBiPath}{\textsc{Exact Bicriteria $k$-Path}\xspace}
\newcommand{\eps}{\ensuremath{\varepsilon}}
\newcommand{\poly}{\textup{poly}}

\newtheorem{oq}{Open Question}

\begin{document}

\title{SETH-Based Lower Bounds for\\ Subset Sum and Bicriteria Path}

\author{
Amir Abboud\inst{1} \and
Karl Bringmann\inst{2} \and
Danny Hermelin \inst{3}  \and Dvir Shabtay \inst{3}}
\institute{
Department of Computer Science,\\ Stanford University, CA, USA\\
\href{mailto:abboud@cs.stanford.edu}{abboud@cs.stanford.edu}
\and
Max Planck Institute for Informatics,\\ Saarland Informatics Campus, Germany\\
\href{mailto:kbringma@mpi-inf.mpg.de}{kbringma@mpi-inf.mpg.de}
\and
Department of Industrial Engineering and Management,\\ Ben-Gurion University, Israel\\
\href{mailto:hermelin@bgu.ac.il}{hermelin@bgu.ac.il}, \href{dvirs@bgu.ac.il}{dvirs@bgu.ac.il}
}

\maketitle

\begin{abstract}
% !TEX root = paper.tex

%%%%%%%%%%%%%%%%%%%%%%%%%%%%%%%%%%%%%%%%%%%%%%%%%%%%%%%%%%%%%%%%%
%%%%%% Abstract
%%%%%%%%%%%%%%%%%%%%%%%%%%%%%%%%%%%%%%%%%%%%%%%%%%%%%%%%%%%%%%%%%

%\small\baselineskip=9pt 
\SSUM and $k$-SAT are two of the most extensively studied problems in computer science, and conjectures about their hardness are among the cornerstones of fine-grained complexity.
An important open problem in this area is to base the hardness of one of these problems on the other.

\hspace{10pt} Our main result is a tight reduction from $k$-SAT to \SSUM on dense instances, proving that Bellman's 1962 pseudo-polynomial $O^{*}(T)$-time algorithm for \SSUM on $n$ numbers and target~$T$ cannot be improved to time $T^{1-\eps}\cdot 2^{o(n)}$ for any $\eps>0$, unless the Strong Exponential Time Hypothesis (SETH) fails.
%This is one of the strongest known connections between any two of the core problems of fine-grained complexity.

\hspace{10pt} As a corollary, we prove a ``Direct-OR" theorem for \SSUM under SETH, offering a new tool for proving conditional lower bounds:
It is now possible to assume that deciding whether one out of $N$ given instances of \SSUM is a YES instance requires time $(N T)^{1-o(1)}$.
As an application of this corollary, we prove a tight SETH-based lower bound for the classical \bipath problem, which is extensively studied in Operations Research. We separate its complexity from that of \SSUM:
On graphs with $m$ edges and edge lengths bounded by $L$, we show that the $O(Lm)$ pseudo-polynomial time algorithm by Joksch from 1966 cannot be improved to $\tilde{O}(L+m)$, in contrast to a recent improvement for Subset Sum (Bringmann, SODA 2017).

%In this paper we consider the \textsc{Bicriteria $s,t$-Path} problem: Given a directed graph with lengths and costs on the edges, find a path from $s$ to $t$ in which the total length is less than a given length budget $L$ \emph{and} the total cost is less than a given cost budget $C$. This is the decision version of a classic NP-hard ``restricted-path" optimization problem that naturally arises in many applications. The problem has been extensively studied in the literature, mostly in the communication and operations research communities.
%
%\hspace{15pt} In this paper we present a  a new $\tilde{O}(n^{\lceil(k-1)/2\rceil})$ time algorithm for \textsc{Bicriteria $s,t$-Path}, where $k$ denotes the number of internal vertices in a solution path, and $\tilde{O}(\cdot)$ suppresses poly-logarithmic factors in $L+C$. We also prove that any improvement by an $n^{\varepsilon}$ factor, for any $\varepsilon>0$, on the running time of this algorithm would refute a conjecture on the complexity of the \textsc{$k$-Sum} problem. Furthermore, we show that the problem can be solved in either $O(n^{\lambda+2})$ or $O(n^{\chi+2})$ time, where $\lambda$ and $\chi$ respectively denote the total number of different edge-lengths and edge-costs in the network, and prove that these bounds cannot be substantially improved under the Exponential Time Hypothesis. In the final part of our paper, we also discuss some implications of our results to the parameterized complexity of the problem. 
\end{abstract}

\section{Introduction}
\label{sec: introduction}
% !TEX root = paper.tex

%%%%%%%%%%%%%%%%%%%%%%%%%%%%%%%%%%%%%%%%%%%%%%%%%%%%%%%%%%%%%%%%%
%%%%%% Section: Introduction
%%%%%%%%%%%%%%%%%%%%%%%%%%%%%%%%%%%%%%%%%%%%%%%%%%%%%%%%%%%%%%%%%

The field of fine-grained complexity is anchored around certain hypotheses about the exact time complexity of a small set of core problems.
Due to dozens of reductions, we now know that the current algorithms for many important problems are optimal unless breakthrough algorithms for the core problems exist.
A central challenge in this field is to understand the connections and relative difficulties among these core problems.
In this work, we discover a new connection between two core problems: a tight reduction from {\sc $k$-SAT} to \SSUM.

In the first part of the introduction we discuss this new reduction and how it affects the landscape of fine-grained complexity.
Then, in Section~\ref{subseq:booleansum}, we highlight a corollary of this reduction which gives a new tool for proving conditional lower bounds.
As an application, in Section~\ref{subseq:bipath}, we prove the first tight bounds for the classical \bipath problem from Operations Research.

\paragraph{\bf Subset Sum.}
%(explain why Subset Sum is the most important problem in the world!)
\SSUM is a fundamental problem in computer science.
Its most basic form is the following: Given $n$ integers $x_1,\ldots,x_n \in \mathbb{N}$, and a target value $T \in \mathbb{N}$, decide whether there is a subset of the numbers that sums to $T$.
%A closely related problem is Knapsack\footnote{Cryptographers usually refer to \SSUM as Knapsack, while in the Algorithms community Knapsack is the problem of picking items of maximum total value and weight upper bounded by some constraint $T$.};
%the two are runtime-equivalent \cite{mfcs12} up to $poly(n)$ factors, which can be ignored for now but will be discussed later on.
The two most classical algorithms for the problem are the pseudo-polynomial $O(T n)$ algorithm using dynamic programming \cite{bellman1957dynamic}, and the $O(2^{n/2} \cdot \textup{poly}(n,\log{T}))$ algorithm via ``meet-in-the-middle" \cite{HS74}.
A central open question in Exact Algorithms \cite{Woe08} is whether faster algorithms exist, \emph{e.g.}, can we combine the two approaches to get a $T^{1/2}\cdot n^{O(1)}$ time algorithm? Such a bound was recently found in a Merlin-Arthur setting \cite{Ned17}.

\begin{oq}
\label{oq1}
Is \SSUM in time $T^{1-\eps} \cdot 2^{o(n)}$ or $2^{(1-\eps)\frac{n}{2}} \cdot T^{o(1)}$, for some $\eps>0$?
\end{oq}

The status of \SSUM as a major problem has been established due to many applications, deep connections to other fields, and educational value.
The $O(Tn)$ algorithm from 1957 is an illuminating example of dynamic programming that is taught in most undergraduate algorithms courses, and the NP-hardness proof (from Karp's original paper \cite{Karp72}) is a prominent example of a reduction to a problem on numbers.
Interestingly, one of the earliest cryptosystems by Merkle and Hellman was based on \SSUM \cite{MH78}, and was later extended to a host of \emph{Knapsack-type} cryptosystems\footnote{Cryptographers usually refer to \SSUM as Knapsack.} (see \cite{Sha84,BO88,Odl90,CR88,IN96} and the references therein).

The version of \SSUM where we ask for $k$ numbers that sum to zero (the \kSUM problem) is conjectured to have $n^{\lceil k/2 \rceil \pm o(1)}$ time complexity.
Most famously, the $k=3$ case is the \textsc{3-Sum} conjecture highlighted in the seminal work of Gajentaan and Overmars~\cite{GO95}.
It has been shown that this problem lies at the core and captures the difficulty of dozens of problems in computational geometry. Searching in Google Scholar for ``3sum-hard" reveals more than 250 papers (see~\cite{King04} for an incomplete list).
More recently, these conjectures have become even more prominent as core problems in fine-grained complexity since their interesting consequences have expanded beyond geometry into purely combinatorial problems \cite{Pat10,VW09,B+13,JV13,ACLL14,AV14,KPP14,AVY15,GKLP16,JV13}.
Note that \kSUM inherits its hardness from \SSUM, by a simple reduction that partitions items into $k$ groups: To answer Open Question~\ref{oq1} positively it is enough to solve \kSUM in time $T^{1-\eps}\cdot n^{o(k)}$ (for all $k$) or in time $n^{k/2-\eps}\cdot T^{o(1)}$ (for any $k$).

Entire books \cite{MT90book,KPP04book} are dedicated to the algorithmic approaches that have been used to attack \SSUM throughout many decades, and, quite astonishingly, major algorithmic advances are still being discovered in our days, \emph{e.g.}, \cite{LN10,HGJ10,BCJ11,DKS12,Austrin13,Wang14,GP14,GS15,CL15,AKKN15,AKKN16,LWWW16,Freund17,BGNV17,Ned17,KoiliarisX17,Bring17}, not to mention the recent developments on generalized versions (see \cite{BCILOS16}) and other computational models (see \cite{Viola15,CIO16}).
At STOC'17 an algorithm was presented that beats the trivial $2^n$ bound while using polynomial space, under certain assumptions on access to random bits~\cite{BGNV17}.
At SODA'17 we have seen the first improvements (beyond log factors \cite{pisinger2003dynamic}) over the $O(Tn)$ algorithm, reducing the bound to $\tilde{O}(T+n)$ \cite{KoiliarisX17,Bring17}.
And a few years earlier, a surprising result celebrated by cryptographers \cite{HGJ10,BCJ11} showed that $2^{0.499}$ algorithms are possible on random instances.
All this progress leads to the feeling that a positive resolution to Open Question 1 might be just around the corner.

\paragraph{\bf SETH}
%change k to d?
$k$-SAT is an equally fundamental problem (if not more) but of a Boolean rather than numerical nature, where we are given a $k$-CNF formula on $n$ variables and $m$ clauses, and the task is to decide whether it is satisfiable.
All known algorithms have a running time of the form $O(2^{(1-c/k) n})$ for some constant $c>0$ \cite{PPSZ05,DH09,AWY15}, and the Strong Exponential Time Hypothesis (SETH) of Impagliazzo and Paturi \cite{IP2001,IPZ2001,CIP06} states that no $O(2^{(1-\eps)n})$ time algorithms are possible for $k$-SAT, for some $\eps>0$ independent of $k$.
Refuting SETH implies advances in circuit complexity \cite{JMV15}, and is known to be impossible with popular techniques like resolution \cite{BI13}.

A seminal paper of Cygan, Dell, Lokshtanov, Marx, Nederlof, Okamoto, Paturi, Saurabh, and Wahlstr\"om \cite{Cygan+16} strives to classify the exact complexity of important NP-hard problems under SETH.
The authors design a large collection of ingenious reductions and conclude that $2^{(1-\eps)n}$ algorithms for problems like Hitting Set, Set Splitting, and Not-All-Equal SAT are impossible under SETH.
Notably, \SSUM is not in this list nor any problem for which the known algorithms are non-trivial (\emph{e.g.}, require dynamic programming).
As the authors point out:
``\emph{Of course, we would also like to show tight connections between SETH and the
optimal growth rates of problems that {\bf do} have non-trivial exact algorithms.}"

Since the work of Cygan et al. \cite{Cygan+16}, SETH has enjoyed great success as a basis for lower bounds in Parameterized Complexity \cite{LMS11} and for problems within P \cite{Vass15}.
Some of the most fundamental problems on strings (\emph{e.g.}, \cite{AVW14,BI15,ABV15a,BK15,BI16,BGL16}), graphs (\emph{e.g.}, \cite{LMS11a,RV13,AV14,GIKW17}), curves (\emph{e.g.}, \cite{Bring14}), vectors \cite{Wil05,WY14,BIL17,BGS17} and trees \cite{ABHVZ16} have been shown to be \emph{SETH-hard}: a small improvement to the running time of these problems would refute SETH.
Despite the remarkable quantity and diversity of these results, we are yet to see a (tight) reduction from SAT to any problem like \SSUM, where the complexity comes from the hardness of analyzing a search space defined by addition of numbers.
In fact, all hardness results for problems of a more \emph{number theoretic} or \emph{additive combinatoric} flavor are based on the conjectured hardness of \SSUM itself.

\medskip

In this paper, we address an important open questions in the field of fine-grained complexity: \emph{Can we prove a tight SETH-based lower bound for \SSUM?}
\medskip

The standard NP-hardness proofs imply loose lower bounds under SETH (in fact, under the weaker ETH) stating that $2^{o(\sqrt{n})}$ algorithms are impossible.
%(Previous SETH lower bounds are much weaker)
A stronger but still loose result
rules out $2^{o(n)} \cdot T^{o(1)}$-time algorithms for \SSUM under ETH~\cite{DBLP:journals/siamdm/JansenLL16,DBLP:journals/mst/BuhrmanLT15}.
Before that, 
Patrascu and Williams \cite{PW10} showed that if we solve \kSUM in $n^{o(k)} \cdot T^{o(1)}$ time, then ETH is false.
These results leave the possibility of $O(T^{0.001})$ algorithms.
While it is open whether such algorithms imply new SAT algorithms, it has been shown that they would imply new algorithms for other famous problems.
Bringmann \cite{Bring17} recently observed that an $O(T^{0.78})$ algorithm for \SSUM implies a new algorithm for $k$-Clique, via a reduction of Abboud, Lewi, and Williams \cite{AbboudLW2014}.
Cygan et al. \cite{Cygan+16} ruled out $O(T^{1-\eps} \poly(n))$ algorithms for \SSUM under the conjecture that the Set Cover problem on $m$ sets over a universe of size $n$ cannot be solved in $O(2^{(1-\eps)n}\cdot \poly(m))$ time. Whether this conjecture can be replaced by the more popular SETH remains a major open question.

%To get a tight lower bound ruling out $O(T^{1-\eps} \poly(n))$ algorithms, while overcoming the difficulty of reducing SETH to NP-hard problems with dynamic programming algorithms, Cygan et al. \cite{Cygan+16} were forced to introduce a new conjecture essentially stating that the Set Cover problem on $m$ sets over a universe of size $n$ cannot be solved in $O(2^{(1-\eps)n}\cdot \poly(m))$ time.
%More than five years after its introduction, this Set Cover Conjecture has not found other major consequences besides this lower bound for \SSUM, and a lower bound for the Steiner Tree problem; other corollaries of the conjecture include tight lower bounds for less famous problems like Connected Vertex Cover, and Set Partitioning~\cite{Cygan+16}.
%Whether this conjecture can be based on or replaced by the much more popular SETH remains a major open question.

\subsection{Main Result}
\label{subseq:main}

We would like to show that SETH implies a negative resolution to Open Question 1.
%  then \SSUM cannot be solved in $O(T^{1-\eps} \cdot 2^{o(n)})$ time, nor in $2^{(1-\eps)n/2} \cdot \log^{O(1)}{T}$ time.
Our main result accomplishes half of this statement, showing a tight reduction from SAT to \SSUM on instances where $T = 2^{\delta n}$, also known as \emph{dense} instances\footnote{The density of an instance is $n/\log_2 \max{x_i}$.}, ruling out $T^{1-\eps}\cdot 2^{o(n)}$ time algorithms under SETH.

\begin{theorem} \label{thm:sethsubsetsum}
  Assuming SETH, for any $\eps > 0$ there exists a $\delta > 0$ such that \SSUM is not in time $O(T^{1-\eps} 2^{\delta n})$, and \kSUM is not in time $O(T^{1-\eps} n^{\delta k})$.
\end{theorem}

Thus, \SSUM is yet another SETH-hard problem. This is certainly a major addition to this list.
This also adds many other problems that have reductions from \SSUM, \emph{e.g.}, the famous Knapsack problem, or from \kSUM (\emph{e.g.}, \cite{Eri99,BIWX11,AL13,CP14,KSR16}).
For some of these problems, to be discussed shortly, this even leads to better lower bounds.

Getting a reduction that also rules out $2^{(1-\eps)n/2} \cdot T^{o(1)}$ algorithms under SETH is still a fascinating open question.
Notably, the strongest possible reduction, ruling out $n^{\lceil k/2 \rceil- \eps} \cdot T^{o(1)}$ algorithms for \kSUM, is provably impossible under the \emph{Nondeterministic SETH} of Carmosino et al. \cite{C+16}, but there is no barrier for an $n^{ k/2- o(1)}$ lower bound.

A substantial technical barrier that we had to overcome when designing our reduction is the fact that there was no clear understanding of what the hard instances of \SSUM should look like.
Significant effort has been put into finding and characterizing the instances of \SSUM and Knapsack that are hard to solve.
This is challenging both from an experimental viewpoint (see the study of Pisinger \cite{Pis05}) and from the worst-case analysis perspective (see the discussion of Austrin et al. \cite{AKKN16}).
Recent breakthroughs refute the common belief that random instances are maximally hard \cite{HGJ10,BCJ11}, and show that better upper bounds are possible for various classes of inputs.
Our reduction is able to generate hard instances by crucially relying on a deep result on the combinatorics of numbers: the existence of dense average-free sets.
A surprising construction of these sets from 1946 due to Behrend \cite{Behrend1946} (see also \cite{Elkin10,OB11}) has already lead to breakthroughs in various areas of theoretical computer science \cite{CFL83,CW90,AKFS99,HW03,DV10,AB16}.
These are very non-random-like structures in combinatorics (in particular they are regularly used as counterexamples in mathematics and thus they seem to be extremal and far from random), which allows our instances to bypass the easyness of random inputs.
This leads us to a candidate distribution of hard instances for \SSUM, which could be of independent interest:
Start from hard instances of SAT (\emph{e.g.}, random formulas around the threshold) and map them with our reduction (the obtained distribution over numbers will be highly structured).

Recently, it was shown that the security of certain cryptographic primitives can be based on SETH \cite{BRSV17a,BRSV17b}.
We hope that our SETH-hardness for an already popular problem in cryptography will lead to further interaction between fine-grained complexity and cryptography.
In particular, it would be exciting if our hard instances could be used for a new Knapsack-type cryptosystem. Such schemes tend to be much more computationally efficient than popular schemes like RSA \cite{BO88,Odl90,IN96}, but almost all known ones are not secure (as famously shown by Shamir \cite{Sha84}).
Even more recently, Bennett, Golovnev, and Stephens-Davidowitz \cite{BGS17} proved SETH hardness for another central problem from cryptography, the Closest-Vector-Problem (CVP). While CVP is a harder problem than Subset Sum, their hardness result addresses a different regime of parameters, and rules out $O(2^{(1-\eps)n})$ time algorithms (when the \emph{dimension} is large). %It would be exciting to combine the two techniques and get a completely tight lower bound for CVP.

\subsection{A Direct-OR Theorem for Subset Sum}
\label{subseq:booleansum}

Some readers might find the above result unnecessary: What is the value in a SETH-based lower bound if we already believe the Set Cover Conjecture of Cygan et al.?
The rest of this introduction discusses new lower bound results that, to our knowledge, would not have been possible without our new SETH-based lower bound.
To clarify what we mean, consider the following ``Direct-OR" version of \SSUM:
Given $N$ different and independent instances of \SSUM, each on $n$ numbers and each with a different target $T_i \leq T$, decide whether any of them is a YES instance.
It is natural to expect the time complexity of this problem to be $(N T)^{1-o(1)}$, but how do we formally argue that this is the case?
If we could assume that this holds, it would be a very useful tool for conditional lower bounds (as we show in Section~\ref{subseq:bipath}).
%To obtain tight lower bounds for other important problems (discussed below) we would like to assume that this holds.

Many problems, like SAT, have a simple self-reduction proving that the ``Direct-OR'' version is hard, assuming the problem itself is hard:
To solve a SAT instance on $n$ variables, it is enough to solve $2^x$ instances on $n-x$ variables.
This is typically the case for problems where a brute force algorithm achieves the best known running time up to lower order factors.
But what about \SSUM or Set Cover?
Can we use an algorithm that solves $N$ instances of \SSUM in $O(N^{0.1}\cdot T)$ time to solve \SSUM in $O(T^{1-\eps})$ time?
We cannot prove such statements; however, we can prove that such algorithms would refute SETH.

\begin{corollary} \label{cor:direct}
Assuming SETH, for any $\eps > 0$ and $\gamma > 0$ there exists a
$\delta > 0$ such that no algorithm can solve the OR of $N$ given instances of \SSUM on target values $T_1, \ldots, T_N = O(N^\gamma)$ and at most
$\delta \log N$ numbers each, in total time $O(N^{1 + \gamma - \eps})$.
%
%Assuming SETH, for any pair of constants $N>0$ and $\eps > 0$, there exists a $\delta > 0$ such that we cannot decide if any of $N$ instances of \SSUM, each on at most $\delta \log N$ items each and target values at most $N^q$, is a YES-instance in time $O(N\cdot T^{1-\eps} 2^{\delta n})$.
\end{corollary}

%[we separate Bicriteria Path and Subset Sum!]

\subsection{The Fine-Grained Complexity of Bicriteria Path}
\label{subseq:bipath}

%The main application of Corollary~\ref{cor:direct} that we report in this paper is a tight lower bound for a

%Consider the following scenario: Having just moved into your new country house, you are planning a route to your office that is located far away in the city center. You quickly observe that there is a really fast route to work, but this route is quite expensive due to several toll roads you have to pass along the way. Alternatively, there is a route that does not pass through any toll road at all, yet it is quite long. It thus becomes clear to you that you need to find a route that is satisfactory with respect to both criteria, time and money.

The \bipath problem is the natural bicriteria variant of the classical \textsc{$s,t$-Path} problem where edges have two types of weights and we seek an $s,t$-path which meets given demands on both criteria.  More precisely, we are given a directed graph~$G$ where each edge $e \in E(G)$ is assigned a pair of non-negative integers $\ell(e)$ and $c(e)$, respectively denoting the \emph{length} and \emph{cost} of $e$, and two non-negative integers $L$ and $C$ representing our budgets. The goal is to determine whether there is an $s,t$-path $e_1,\ldots,e_k$ in $G$, between a given source and a target vertex $s,t \in V(G)$, such that $\sum_{i=1}^k \ell(e_i) \leq L$ and $\sum_{i=1}^k c(e_i) \leq C$. %Such a path is called a \emph{feasible} $s,t$-path, and it corresponds to a path in which both criteria are met.

%The \textsc{Bicriteria $s,t$-Path} problem has been extensively studied in the literature, by various research communities, and has many diverse applications in several areas~\cite{GarroppoEtAl10}. The example given above is taken from the domain of transportation networks. In communication networks, an example of an application that has been thoroughly researched is the \emph{quality of service (QoS) routing problem}, where one is required to find a path in the network that satisfies two (or more) QoS constraints such as guaranteed bandwidth, bounded end-to-end delay, very low packet loss, limited interference among wireless links, and so forth~\cite{LorenzOrda98,YounisF03}. There are also several applications for the problem in operation research domains, in particular in the area of  scheduling~\cite{BriskornEtAl10,LeeEtAl09,NaorEtAl07,ShabtayEtAl14}, and in column generation techniques~\cite{HolmbergYuan03,Zimmermann2003}. Additional applications can be found in road traffic management, navigation systems, freight transportation, supply chain management and pipeline distribution systems~\cite{GarroppoEtAl10}.

This natural variant of \textsc{$s,t$-Path} has been extensively studied in the literature, by various research communities, and has many diverse applications in several areas. Most notable of these are perhaps the applications in the area of transportation networks~\cite{GarroppoEtAl10}, and the \emph{quality of service (QoS) routing problem} studied in the context of communication networks~\cite{LorenzOrda98,YounisF03}. There are also several applications for \bipath in Operations Research domains, in particular in the area of  scheduling~\cite{BriskornEtAl10,LeeEtAl09,NaorEtAl07,ShabtayEtAl14}, and in column generation techniques~\cite{HolmbergYuan03,Zimmermann2003}. Additional applications can be found in road traffic management, navigation systems, freight transportation, supply chain management and pipeline distribution systems~\cite{GarroppoEtAl10}.

A simple reduction proves that \textsc{Bicriteria $s,t$-Path} is at least as hard as \SSUM (see Garey and Johnson~\cite{GareyJohnson79}).
In 1966, Joksch~\cite{Joksch66} presented a dynamic programming algorithm with pseudo-polynomial running time $O(L m)$ (or $O(C m)$) on graphs with $m$ edges. Extensions of this classical algorithm appeared in abundance since then, see \emph{e.g.},~\cite{aneja78,Hansen1980,RaithE09} and the various FPTASs for the optimization variant of the problem~\cite{ErgunEtAl02,GarroppoEtAl10,Hassin92,LorenzRaz01,Warburton87}. The reader is referred to the survey by Garroppo et al.~\cite{GarroppoEtAl10} for further results on \bipath.

%A simple reduction proves that \textsc{Bicriteria $s,t$-Path} is harder than \SSUM (see Garey and Johnson~\cite{GareyJohnson79}).
%In 1966, Joksch~\cite{Joksch66} presented a dynamic programming algorithm running in $O(W \cdot m)$ (or $O(L\cdot m)$) time for the problem on graphs on $m$ edges, showing that it is pseudo-polynomial time solvable. Extensions of this algorithm appeared in abundance since then, see \emph{e.g.}~\cite{aneja78,Hansen1980,RaithE09}. The optimization variant of \textsc{Bicriteria $s,t$-Path}, where one optimizes the length of the $s,t$-path given a constraint on its total cost, has been extensively studied from the perspective of approximation algorithms. Warburton~\cite{Warburton87} was the first to present an FPTAS for the problem. Other notable FPTAS's for these problems were later presented by Hassin~\cite{Hassin92}, Lorenz and Raz~\cite{LorenzRaz01}, and Ergun \emph{et al}.~\cite{ErgunEtAl02}. According to our knowledge, the theoretically best FPTAS is due to Lorenz and Raz; it runs in $O(n^4(\lg \lg n + 1/\varepsilon))$ time on $n$ node graphs. We refer the reader to the survey by Garroppo \emph{et al}.~\cite{GarroppoEtAl10} for further results on \textsc{Bicriteria $s,t$-Path}.

Our SETH-based lower bound for \SSUM easily transfers (using, \emph{e.g.}, the reduction in~\cite{GareyJohnson79}) to show that an $O(L^{1-\eps}2^{o(n)})$ time algorithm for \bipath refutes SETH. However, after the $O(Tn)$ algorithm for \SSUM from 1960 was improved last year to $\tilde{O}(T+n)$, it is natural to wonder if the similar $O(Lm)$ algorithm for \bipath from 1966 can also be improved to $\tilde{O}(L+m)$ or even just to $O(Lm^{0.99})$. Such an improvement would be very interesting since the pseudo-polynomial algorithm is commonly used in practice, and since it would speed up the running time of the approximation algorithms. We prove that \bipath is in fact a harder problem than \SSUM, and an improved algorithm would refute SETH. The main application of Corollary~\ref{cor:direct} that we report in this paper is a tight SETH-based lower bound for \bipath, which (conditionally) separates the time complexity of \bipath and \SSUM.
\begin{theorem}
\label{thm:cspath}%
Assuming SETH, for any $\eps > 0$ and $\gamma > 0$ no algorithm solves \bipath on sparse $n$-vertex graphs and budgets $L, C = \Theta(n^\gamma)$ in time $O(n^{1+\gamma-\eps})$.
\end{theorem}
Intuitively, our reduction shows how a single instance of \bipath can simulate multiple instances of \SSUM and solve the ``Direct-OR" version of it. %Whether such lower bounds are possible for Knapsack, or whether there are $O(T+n)$ algorithms is an interesting open question, also asked by Bringmann and Cygan et al.]

Our second application of Corollary~\ref{cor:direct} concerns the number of different edge-lengths and/or edge-costs in our given input graph.
Let $\lambda$ denote the former parameter, and $\chi$ denote the latter. Note that $\lambda$ and $\chi$ are different from $L$ and $C$, and each can be quite small in comparison to the size of the entire input. In fact, in many of the scheduling applications for \bipath discussed above it is natural to assume that one of these is quite small. We present a SETH-based lower bound that almost matches the $O(n^{\min\{\lambda,\chi\}+2})$ upper bound for the problem.
\begin{theorem}
\label{thm: lambdachi}%
\textsc{Bicriteria $s,t$-Path} can be solved in $O(n^{\min\{\lambda,\chi\}+2})$ time. Moreover, assuming SETH, for any constants $\lambda, \chi \ge 2$ and $\varepsilon>0$, there is no $O(n^{\min\{\lambda,\chi\} - 1 -\eps})$ time algorithm for the problem.
\end{theorem}

Finally, we consider the case where we are searching for a path that uses only $k$ internal vertices. This parameter is naturally small in comparison to the total input length in several applications of \bipath, for example the packet routing application discussed above. We show that this problem is equivalent to the \kSUM problem, up to logarithmic factors.
For this, we consider an intermediate exact variant of \bipath, the {\sc Zero-Weight-$k$-Path} problem, and utilize the known bounds for this variant to obtain the first improvement over the $O(n^k)$-time brute-force algorithm, as well as a matching lower bound.
\begin{theorem}
\label{thm: k}
\textsc{Bicriteria $s,t$-Path} can be solved in $\tilde{O}(n^{\lceil(k+1)/2\rceil})$ time. Moreover, for any $\varepsilon>0$, there is no $\tilde{O}(n^{\lceil(k+1)/2\rceil - \varepsilon})$-time algorithm for the problem, unless \kSUM has an $\tilde{O}(n^{\lceil k/2\rceil - \varepsilon})$-time algorithm.
\end{theorem}

\section{Preliminaries}
\label{sec: preliminaries}
% !TEX root = paper.tex

%%%%%%%%%%%%%%%%%%%%%%%%%%%%%%%%%%%%%%%%%%%%%%%%%%%%%%%%%%%%%%%%%
%%%%%% Section: Preliminaries
%%%%%%%%%%%%%%%%%%%%%%%%%%%%%%%%%%%%%%%%%%%%%%%%%%%%%%%%%%%%%%%%%

For a fixed integer $p$, we let $[p]$ denote the set of integers $\{1,\ldots,p\}$. All graphs in this paper are, unless otherwise stated, simple, directed, and without self-loops. We use standard graph theoretic notation, \emph{e.g.}, for a graph $G$ we let $V(G)$ and $E(G)$ denote the set of vertices and edges of $G$, respectively.
Throughout the paper, we use the $O^*(\cdot)$ and $\tilde{O}(\cdot)$ notations to suppress polynomial and logarithmic factors.

\paragraph{Hardness Assumptions:} The Exponential Time Hypothesis (ETH) and its strong variant (SETH) are conjectures about running time of any algorithm for the \kSAT problem: Given a boolean CNF formula $\phi$, where each clause has at most $k$ literals, determine whether $\phi$ has a satisfying assignment. Let $s_k = \inf\{ \delta : \text{\kSAT can be solved in } O^*(2^{\delta n}) \text{ time} \}$. The Exponential Time Hypothesis, as stated by Impagliazzo, Paturi and Zane~\cite{IPZ2001}, is the conjecture that $s_3 > 0$. It is known that $s_3 > 0$ if and only if there is a $k \geq 3$ such that $s_k > 0$~\cite{IPZ2001}, and that if ETH is true, the sequence $\{s_k\}^\infty_{k=1}$ increases infinitely often~\cite{IP2001}. The Strong Exponential Time Hypothesis, coined by Impagliazzo and Paturi~\cite{CalabroIP09,IP2001}, is the conjecture that $\lim_{k \to \infty} s_k = 1$. In our terms, this can be stated in the following more convenient manner:
\begin{conjecture}
\label{conj:SETH}%
For any $\eps > 0$ there exists $k \ge 3$ such that \kSAT on $n$ variables cannot be solved in time $O(2^{(1-\eps)n})$.
\end{conjecture}

We use the following standard tool by Impagliazzo, Paturi and Zane:
\begin{lemma}[Sparsification Lemma~\cite{IPZ2001}]
\label{lem:sparsification}%
For any $\eps > 0$ and $k \ge 3$, there exists $c_{k,\eps} > 0$ and an algorithm that, given a \kSAT instance $\phi$ on $n$ variables, computes \kSAT instances $\phi_1, \ldots, \phi_\ell$ with $\ell \le 2^{\eps n}$ such that $\phi$ is satisfiable if and only if at least one $\phi_i$ is satisfiable. Moreover, each $\phi_i$ has $n$ variables, each variable in $\phi_i$ appears in at most $c_{k,\eps}$ clauses, and the algorithm runs in time $\poly(n) 2^{\eps n}$.
\end{lemma}

\paragraph{The \kSUM Problem:} In \kSUM we are given sets $Z_1,\ldots,Z_k$ of non-negative integers and a target~$T$, and we want to decide whether there are $z_1 \in Z_1, \ldots, z_k \in Z_k$ such that $z_1 + \ldots + z_k = T$. This problem can be solved in time $O(n^{\lceil k/2 \rceil})$~\cite{HS74}, and it is somewhat standard by now to assume that this is essentially the best possible~\cite{AL13}. This assumption, which generalizes the more popular assumption of the $k=3$ case~\cite{GO95,Pat10}, remains believable despite recent algorithmic progress~\cite{Austrin13,BDP05,CL15,GP14,Wang14}.
\begin{conjecture}
\label{conj:ksum}%
\kSUM cannot be solved in time $\tilde O(n^{\lceil k/2 \rceil - \eps})$ for any $\eps > 0$ and $k \ge 3$. %, where the $\tilde{O}(\cdot)$ notation suppresses poly-logarithmic factors in $t$.
\end{conjecture}

\section{From SAT to Subset Sum}
\label{sec: sethsubsetsum}
% !TEX root = paper.tex

In this section we present our main result, the hardness of \SSUM and \kSUM under SETH. Our reduction goes through three main steps: We start with a \kSAT formula $\phi$ that is the input to our reduction. This formula is then reduced to subexponentially many Constraint Satisfaction Problems (CSP) with a restricted structure. The main technical part is then to reduce these CSP instances to equivalent \SSUM instances. The last part of our construction, reducing \SSUM to \kSUM, is rather standard. In the final part of the section we provide a proof for Corollary~\ref{cor:direct}, showing that \SSUM admits the ``Direct-OR" property discussed in Section~\ref{subseq:booleansum}.

\subsection{From \boldmath{$k$}-SAT to Structured CSP}
%\input{sat2csp}%
% !TEX root = paper.tex

We first present a reduction from \kSAT to certain structured instances of Constraint Satisfaction Problems (CSP). This is a standard combination of the Sparsification Lemma with well-known tricks.

\begin{lemma} \label{lem:csp}
Given a \kSAT instance $\phi$ on $n$ variables and $m$ clauses, for any $\eps > 0$ and $a \ge 1$ in time $\poly(n) 2^{\eps n}$ we can compute CSP instances $\psi_1,\ldots,\psi_\ell$, with $\ell \le 2^{\eps n}$, such that $\phi$ is satisfiable if and only if some $\psi_i$ is satisfiable. Each $\psi_i$ has $\hat n = \lceil n/a \rceil$ variables over universe $[2^a]$ and $\hat m = \lceil n/a \rceil$ constraints. Each variable is contained in at most $\hat c_{k,\eps} \cdot a$ constraints, and each constraint contains at most $\hat c_{k,\eps} \cdot a$ variables, for some constant $\hat c_{k,\eps}$ depending only on $k$ and $\eps$.
\end{lemma}

\begin{proof}
Let $\phi$ be an instance of \kSAT with $n$ variables and $m$ clauses. We start by invoking the Sparsification Lemma (Lemma~\ref{lem:sparsification}). This yields \kSAT instances $\phi_1,\ldots,\phi_\ell$ with $\ell \le 2^{\eps n}$ such that $\phi$ is satisfiable if and only if some $\phi_i$ is satisfiable, and where each $\phi_i$ has $n$ variables, and each variable in $\phi_i$ appears in at most $c_{k,\eps}$ clauses of $\phi_i$, for some constant $c_{k,\eps}$. In particular, the number of clauses is at most $c_{k,\eps} n$. %(The constant $c_{k,\eps}$ is known to be $(k/\eps)^{O(k)}$ \karl{Check!}.)

%Fix a formula $\phi_i$ with $n$ variables and $\tilde m \le c_{k,\eps} \cdot n$ clauses.
%We next ensure that each variable appears in only a bounded number of clauses\footnote{This might hold for sparsified \kSAT instances in general, but since we did not find a proof, we instead ensure this property ``by hand''.}. To this end, for any variable $x_j$ appearing in $d_j$ clauses we add $\lfloor d_j / \tau \rfloor$ copies of $x_j$, additional to the original variable $x_j$, where we choose $\tau := k \tilde m / \eps n$.
%This procedure introduces at most $\eps n$ new variables, since as $\phi_i$ is a \kSAT instance we have $\sum_j d_j \le k \cdot \tilde m$, and thus $\sum_j \lfloor d_j / \tau \rfloor \le k \tilde m / \tau = \eps n$.
%We distribute the clauses containing variable $x_j$ over these copies, so that each copy appears in at most $\tau$ clauses.
%We ensure coherence of the copies by introducing two clauses for each new copy, that encode equality with the previous copy. Thus, each variable is in at most $\tau + 4$ clauses.
%In total, this yields an equivalent \kSAT instance $\phi_i'$ with $n' = (1+\eps) n$ variables and $m' \le \tilde  m + 2 \eps n \le (c_{k,\eps} + 2 \eps) n \le 3 c_{k,\eps} n$ clauses, where each variable appears in at most $\tau + 4 = k \tilde m / \eps n + 4 \le 5 k c_{k,\eps} / \eps$ clauses.

We combine multiple variables to a super-variable and multiple clauses to a super-constraint, which yields a certain structured CSP. Specifically, let $a \ge 1$, and partition the variables into $\lceil n/a \rceil$ blocks of length at most $a$. We replace each block of at most $a$ variables by one super-variable over universe $[2^a]$. Similarly, we partition the clauses into $\lceil n/a \rceil$ blocks, each containing at most $\gamma := a c_{k,\eps}$ clauses. We replace each block of $\gamma' \leq \gamma$ clauses $C_1,\ldots,C_{\gamma'}$ by one super-constraint $C$ that depends on all super-variables containing variables appearing in $C_1,\ldots,C_{\gamma'}$.

Clearly, the resulting CSP $\psi_i$ is equivalent to $\phi_i$.
%Since each clause contains at most $k$ variables, each super-constraint contains at most $k \cdot \gamma = k \,a\, m'/n' \le 3 \, k \,a \,c_{k,\eps}$ super-variables.
%Moreover, since in $\phi_i'$ each variable appears in at most $5k  c_{k,\eps}/\eps$ clauses, each super-variable appears in at most $2k \, a\, c_{k,\eps}/\eps$ clauses.
%This yields an equivalent CSP $\psi_i$ on $n'' = n'/a \le (1+\eps) n/a$ variables and $m'' = n''$ constraints. Each variable is contained in at most $c'_{k,\eps} a$ clauses. Indeed, in $\phi_i$ each variable appears in at most $2k  c_{k,\eps}/\eps$ clauses and we combine $a$ variables, so the statement holds with $c'_{k,\eps} := 2k  c_{k,\eps}/\eps$.
%Moreover, each constraint contains at most $c'_{k,\eps} a$ variables. Indeed, each clause in $\phi_i$ contains at most $k$ variables, and each constraint consists of $\gamma = a m'/n' \le a c_{k,\eps}$ clauses.
Since each variable appears in at most $c_{k,\eps}$ clauses in $\phi_i$, and we combine at most $a$ variables to obtain a variable of $\psi_i$, each variable appears in $\psi_i$ in at most $a c_{k,\eps}$ constraints. Similarly, each clause in $\phi_i$ contains at most $k$ variables, and each super-constraint consists of at most $\gamma = a c_{k,\eps}$ clauses, so each super-constraint contains at most $\hat c_{k,\eps} a$ variables for $\hat c_{k,\eps} = k c_{k,\eps}$. This finishes the proof.
\end{proof}

\subsection{From Structured CSP to Subset Sum}

Next we reduce to \SSUM. Specifically, we show the following.

\begin{theorem} \label{thm:redsatsubsetsum}
For any $\eps > 0$, given a \kSAT instance $\phi$ on $n$ variables we can in time $\poly(n) 2^{\eps n}$ construct $2^{\eps n}$ instances of \SSUM on at most $\tilde c_{k,\eps} n$ items and a target value bounded by $2^{(1+2\eps)n}$ such that $\phi$ is satisfiable iff at least one of the \SSUM instances is a YES-instance. Here $\tilde c_{k,\eps}$ is a constant depending only on $k$ and $\eps$.
\end{theorem}

%It remains to prove Theorem~\ref{thm:redsatsubsetsum}.
As discussed in Section~\ref{subseq:main}, our reduction crucially relies on a construction of average-free sets.
For any $k \ge 2$, a set $S$ of integers is \emph{$k$-average-free} iff for all $k'\leq k$ and (not necessarily distinct) $x_1,\ldots,x_{k'+1} \in S$ with $x_1 + \ldots + x_{k'} = k' \cdot x_{k'+1}$ we have $x_1 = \ldots = x_{k'+1}$. A surprising construction by Behrend~\cite{Behrend1946} has been slightly adapted in~\cite{AbboudLW2014}, showing the following.
\begin{lemma}
\label{lem:sumfreeset}%
There exists a universal constant $c > 0$ such that, given $\eps \in (0,1)$, $k \ge 2$, and $n \ge 1$, a $k$-average-free set $S$ of size $n$ with $S \subset [0,k^{c/\eps} n^{1+\eps}]$ can be constructed in $\poly(n)$ time.
\end{lemma}

While it seems natural to use this lemma when working with an additive problem like \SSUM, we are only aware of very few uses of this result in conditional lower bounds~\cite{AbboudLW2014,DBLP:journals/siamcomp/FominGLS14,DBLP:journals/jcss/JansenKMS13}.
One example is a reduction from \textsc{$k$-Clique} to $k^2$-\textsc{Sum} on numbers in $n^{k+o(1)}$ \cite{AbboudLW2014}.
Our result can be viewed as a significant boosting of this reduction, where we exploit the power of \SSUM further.
Morally, \textsc{$k$-Clique} is like \textsc{Max-$2$-SAT}, since faster algorithms for \textsc{$k$-Clique} imply faster algorithms for  \textsc{Max-$2$-SAT}~\cite{Wil05}.
We show that even  \textsc{Max-$d$-SAT}, for any $d$, can be reduced to \kSUM, which corresponds to a reduction from \textsc{Clique} on \emph{hyper-graphs} to \kSUM.

\begin{proof}[Proof of Theorem~\ref{thm:redsatsubsetsum}]
We let $a \ge 1$ be a sufficiently large constant depending only on $k$ and $\eps$.
We need a $\lambda$-average-free set, with $\lambda := \hat c_{k,\eps} a$, where $\hat c_{k,\eps}$ is the constant from Lemma~\ref{lem:csp}.
Lemma~\ref{lem:sumfreeset} yields a $\lambda$-average-free set $S$ of size $2^a$ consisting of non-negative integers bounded by $B := \lambda^{c/\eps} (2^a)^{1+\eps}$, for some universal constant $c>0$.
%Let $B := 2^{(1+\eps) a}$. Lemma~\ref{lem:apfreeset} shows that the set $[B]$ contains a $\lambda$-AP free set of size $B / 2^{O(\sqrt{\log B \cdot \log \lambda})}$. This size is larger than $2^a$ for sufficiently large $a$, since $2^{O(\sqrt{\log B \cdot \log \lambda})} = 2^{O(\sqrt{ a \cdot \log (\hat c_{k,\eps} a)})}$ grows slower than $B / 2^a = 2^{\eps a}$. In fact, one can choose $a = \Theta(\log(\hat c_{k,\eps}) / \eps^2)$ with sufficiently large hidden constant.
%Under this condition, the set $[B]$ contains a $\lambda$-AP free set $S$ of size $2^a$, and
We let $f \colon [2^a] \to S$ be any injective function. Note that since $a$ and $B$ are constants constructing $f$ takes constant time.

Run Lemma~\ref{lem:csp} to obtain CSP instances $\psi_1, \ldots, \psi_\ell$ with $\ell \le 2^{\eps n}$, each with $\hat n = \lceil n /a \rceil$ variables over universe $[2^a]$ and $\hat m = \hat n$ constraints, such that each variable is contained in at most $\lambda$ constraints and each constraint contains at most $\lambda$ variables.
Fix a CSP $\psi = \psi_i$.
We create an instance $(Z,T)$ of \SSUM, i.e., a set $Z$ of positive integers and a target value $T$.
We define these integers by describing blocks of their bits, from highest to lowest. (The items in $Z$ are naturally partitioned, as for each variable $x$ of $\psi$ there will be $O_{k,\eps}(1)$ items \emph{of type~$x$}, and for each clause $C$ of $\psi$  there will be $O_{k,\eps}(1)$ items \emph{of type $C$}.)

We first ensure that any correct solution picks exactly one item of each type. To this end, we start with a block of $O(\log \hat n)$ bits where each item has value 1, and the target value is $\hat n + \hat m$, which ensures that we pick exactly $\hat n + \hat m$ items.
This is followed by $O(\log \hat n)$ many 0-bits to avoid overflow from the lower bits (we will have $O_{k,\eps}(\hat n)$ items overall). In the following $\hat n + \hat m$ bits, each position is associated to one type, and each item of that type has a 1 at this position and 0s at all other positions. The target $T$ has all these bits set to 1. Together, these $O(\log \hat n) + \hat n + \hat m$ bits ensure that we pick exactly one item of each type (since choosing any duplicate type among the $\hat n + \hat m$ picked items leads to two bits that cancel, leaving just one carry bit, and therefore we could not cover all $\hat n + \hat m$ bits in the target $T$). We again add $O(\log \hat n)$ many 0-bits to avoid overflow from the lower bits.

The remaining $\hat n$ blocks of bits correspond to the variables of $\psi$.
For each variable we have a block consisting of $\lceil \log (2 \lambda B+1) \rceil = \log B + \log \lambda + O(1)$ bits. The target number $T$ has bits forming the number $\lambda B$ in each block of each variable.

Now we describe the items of type $x$, where $x$ is a variable. For each assignment $\alpha \in [2^a]$ of $x$, there is an item $z(x,\alpha)$ of type $x$. In the block corresponding to variable $x$, the bits of $z(x,\alpha)$ form the number  $\lambda B - d(x) \cdot f(\alpha)$, where $d(x)$ is the number of clauses containing $x$. In all blocks corresponding to other variables, the bits of $z(x,\alpha)$ are 0.

Next we describe the items of type $C$, where $C$ is a constraint. Let $x_1,\ldots,x_s$ be the variables that are contained in $C$. For any assignment $\alpha_1,\ldots,\alpha_s \in [2^a]$ of $x_1,\ldots,x_s$ that satisfies the clause $C$, there is an item $z(C,\alpha_1,\ldots,\alpha_s)$ of type $C$. In the block corresponding to variable $x_i$ the bits of $z(C,\alpha_1,\ldots,\alpha_s)$ form the number $f(\alpha_i)$, for any $1 \le i \le s$. In all blocks corresponding to other variables, the bits are~0.

\paragraph{Example:}
Suppose $a = 1$ and $\hat c_{k,\eps}=2$, and consider a CSP with variables $x_1,x_2,x_3$ over the universe $[2^a] = \{1,2\}$, and constraints $C_1 = (x_1 = x_2)$, $C_2 = (x_2 \neq x_3)$, and $C_3 = (x_1 =1 \Rightarrow x_3 =1)$. Note that $\lambda = 2$. We construct the 2-average-free set $S=\{1,2\}$; in particular, we may set $B = 2$, and use the injective mapping $f:[2^a] \to S$ defined by $f(x)=x$. The following items correspond to the CSP variables (the $|$-symbols mark block boundaries and have no other meaning):
\begin{align*}
z(x_1,1)   &=    1|000|100000|000|0010|0000|0000| \\
z(x_1,2)   &=    1|000|100000|000|0000|0000|0000| \\
z(x_2,1)   &=    1|000|010000|000|0000|0010|0000| \\
z(x_2,2)   &=    1|000|010000|000|0000|0000|0000| \\
z(x_3,1)   &=    1|000|001000|000|0000|0000|0010| \\
z(x_3,2)   &=    1|000|001000|000|0000|0000|0000|
\end{align*}
And the following items correspond to the constraints:
\begin{align*}
z(C_1,1,1) &=    1|000|000100|000|0001|0001|0000| \\
z(C_1,2,2) &=    1|000|000100|000|0010|0010|0000| \\
z(C_2,1,2) &=    1|000|000010|000|0000|0001|0010| \\
z(C_2,2,1) &=    1|000|000010|000|0000|0010|0001| \\
z(C_3,1,1) &=    1|000|000001|000|0001|0000|0001| \\
z(C_3,2,1) &=    1|000|000001|000|0010|0000|0001| \\
z(C_3,2,2) &=    1|000|000001|000|0010|0000|0010|
\end{align*}
We set the target to 
\[T =  110|000|111111|000|0100|0100|0100|.\] 
One can readily verify that $T$ sums up to $z(x_1,2) + z(x_2,2) + z(x_2,1) + z(C_1,2,2) + z(C_2,2,1) + z(C_3,2,1)$, and that no other subset sums up to $T$. That is, the subsets summing to $T$ are
in one-to-one correspondence to the satisfying assignments of the CSP.

\paragraph{Correctness:}
Recall that the first $O(\log \hat n) + \hat n + \hat m$ bits ensure that we pick exactly one item of each type. Consider any variable $x$ and the corresponding block of bits. The item of type $x$ picks an assignment $\alpha$, resulting in the number $\lambda B - d(x) \cdot f(\alpha)$, where $d(x)$ is the degree of $x$. The $d(x)$ constraints containing $x$ pick assignments $\alpha_1, \ldots, \alpha_{d(x)}$ and contribute $f(\alpha_1) + \ldots + f(\alpha_{d(x)})$. Hence, the total contribution in the block is
\[ f(\alpha_1) + \ldots + f(\alpha_{d(x)}) - d(x) \cdot f(\alpha) + \lambda B, \]
where $d(x) \le \lambda$. Since $f$ maps to a $\lambda$-average-free set, we can only obtain the target $\lambda B$ if $f(\alpha_1) = \ldots  = f(\alpha_{d(x)}) = f(\alpha)$. Since $f$ is injective, this shows that any correct solution picks a coherent assignment $\alpha$ for variable $x$. Finally, this coherent choice of assignments for all variables satisfies all clauses, since clause items only exist for assignments satisfying the clause. Hence, we obtain an equivalent \SSUM instance.

Note that the length of blocks corresponding to variables is set so that there are no carries between blocks, which is necessary for the above argument. Indeed, the degree $d(x)$ of any variable~$x$ is at most $\lambda$, so the clauses containing $x$ can contribute at most $\lambda \cdot B$ to its block, while the item of type $x$ also contributes $0 \le \lambda B - d(x) \cdot f(\alpha) \le \lambda B$, which gives a number in $[0,2 \lambda B]$.

\paragraph{Size Bounds:}
Let us count the number of bits in the constructed numbers. We have $O(\log \hat n) + \hat n + \hat m$ bits from the first part ensuring that we pick one item of each type, and $\hat n \cdot (\log B + \log \lambda + O(1))$ bits from the second part ensuring to pick coherent and satisfying assignments.
This yields
\begin{align*} 
  \log T &= O(\log \hat n) + \hat n + \hat m + \hat n \cdot (\log B + \log \lambda + O(1)).
\end{align*}
We now plug in $B = \lambda^{c/\eps} (2^a)^{1+\eps}$ and $\lambda = \hat c_{k,\eps} a$ and $\hat n = \hat m = \lceil n/a \rceil$ and we bound $\hat n \log B = \big(\frac n a + O(1)\big) \big(\frac c \eps \log \lambda + (1+\eps) a\big) = (1+\eps) n + O_{k,\eps}(n \log (a) / a )$, to obtain
\begin{align*} 
  \log T &= (1 + \eps) n + O_{k,\eps}( n \log(a) /a ), 
\end{align*}
where the hidden constant depends only on $k$ and $\eps$.
Since $(\log a) /a$ tends to 0 for $a \to \infty$, we can choose $a$ sufficiently large, depending on $k$ and $\eps$, to obtain $\log T \le (1+\eps) n + \eps n \le (1+2 \eps) n$. %Hence, $T \le 2^{(1+4 \eps) n}$.

Let us also count the number of constructed items. We have one item for each variable $x$ and each assignment $\alpha \in [2^a]$, amounting to $2^a \hat n \le 2^a n$ items. Moreover, we have one item for each clause $C$ and all assignments $\alpha_1,\ldots,\alpha_s \in [2^a]$ that jointly satisfy the clause $C$, where $s \le \lambda$ is the number of variables contained in $C$. This amounts to up to $2^{a \lambda} \hat m \le 2^{a \lambda} n \le 2^{\hat c_{k,\eps} a^2} n$ items. Note that both factors only depend on $k$ and $\eps$, since $a$ only depends on $k$ and $\eps$. Thus, the number of items is bounded by $\tilde c_{k,\eps} n$, where $\tilde c_{k,\eps}$ only depends on $k$ and $\eps$.

In total, we obtain a reduction that maps an instance $\phi$ of \kSAT on $n$ variables to $2^{\eps n}$ instances of \SSUM with target at most $2^{(1+2 \eps)n}$ on at most $\tilde c_{k,\eps} n$ items. The running time of the reduction is clearly $\poly(n) 2^{\eps n}$.
%Note that our reduction even produces sparse instances of \SSUM, in the sense that the number of items is as small as possible, up to a constant factor depending only on $k$ and $\eps$.
%The hidden constant in the number of items is $2^{a \lambda}$, which by our choices of $\lambda = \hat c_{k,\eps} a$, $\hat c_{k,\eps} = k c_{k,\eps} / \eps$, and $a = \Theta(\log(\hat c_{k,\eps}) / \eps^2)$ evaluates to $2^{\textup{poly}(k/\eps) \cdot c_{k,\eps}}$, where $c_{k,\eps} = (k/\eps)^{O(k)}$ \karl{Check!} is the constant of the sparsification lemma.
\end{proof}

Our main result (Theorem~\ref{thm:sethsubsetsum}) now follows.

\begin{proof}[Proof of Theorem~\ref{thm:sethsubsetsum}]
\emph{Subset Sum:}
For any $\eps > 0$ set $\eps' := \eps/5$ and let $k$ be sufficiently large so that \kSAT has no $O(2^{(1-\eps') n})$ algorithm; this exists assuming SETH. Set $\delta := \eps' / \tilde c_{k,\eps'}$, where $\tilde c_{k,\eps'}$ is the constant from Theorem~\ref{thm:redsatsubsetsum}. Now assume that \SSUM can be solved in time $O(T^{1-\eps} 2^{\delta n})$. We show that this contradicts SETH. Let $\phi$ be a \kSAT instance on $n$ variables, and run Theorem~\ref{thm:redsatsubsetsum} with $\eps'$ to obtain $2^{\eps' n}$ instances of \SSUM on at most $\tilde c_{k,\eps'} n$ items and target at most $2^{(1+2\eps')n}$. Using the assumed $O(T^{1-\eps} 2^{\delta n})$ algorithm on each \SSUM instance, yields a total time for \kSAT of
\begin{align*} 
  &O\big(\poly(n) 2^{\eps' n} + 2^{\eps' n} \cdot \big(2^{(1+2\eps')n} \big)^{1-\eps} 2^{\delta \cdot \tilde c_{k,\eps'} n} \big) \\
  &= O\big( 2^{(\eps' + (1+2 \eps')(1-5\eps') + \eps')n} \big) \le O\big( 2^{(1-\eps') n} \big), 
\end{align*}
where we used the definitions of $\eps'$ and $\delta$ as well as $(1+2 \eps')(1-5\eps') \le 1-3\eps'$. This running time contradicts SETH, yielding the lower bound for \SSUM.

\emph{\kSUM:}
The lower bound $O(T^{1-\eps} n^{\delta k})$ for \kSUM now follows easily from the lower bound for \SSUM.
%
%\begin{proof}[of Theorem~\ref{thm:sethsubsetsum} for $k$-SUM]
%Let $\eps > 0$, and let $\delta$ be the constant from Theorem~\ref{thm:sethsubsetsum}, so that \SSUM has no $O(T^{1-\eps} 2^{\delta n})$ algorithm, assuming SETH.
Consider a \SSUM instance $(Z,T)$ on $|Z| = n$ items and target $T$. Partition $Z$ into sets $Z_1, \ldots, Z_k$ of of equal size, up to $\pm 1$. For each set $Z_i$, enumerate all subset sums $S_i$ of $Z_i$, ignoring the subsets summing to larger than $T$. Consider the \kSUM instance $(S_1,\ldots,S_k,T)$, where the task is to pick items $s_i \in S_i$ with $s_1 + \ldots + s_k = T$. Since $|S_i| \le O(2^{n/k})$, an $O(T^{1-\eps} n^{\delta k})$ time algorithm for \kSUM now implies an $O(T^{1-\eps} 2^{\delta n})$ algorithm for \SSUM, thus contradicting SETH.
%This contradicts SETH by Theorem~\ref{thm:sethsubsetsum}.
\end{proof}

%\subsection{From Subset Sum to \boldmath{$k$}-Sum}
%\input{ssum2ksum}

\subsection{Direct-OR Theorem for Subset Sum}
We now provide a proof for Corollary~\ref{cor:direct}. We show that deciding whether at least one of $N$ given instances of \SSUM is a YES-instance requires time $(N T)^{1-o(1)}$, where $T$ is a common upper bound on the target. Here we crucially use our reduction from \kSAT to \SSUM, since the former has an easy self-reduction allowing us to tightly reduce one instance to multiple subinstances, while such a self-reduction is not known for \SSUM.

\begin{proof}[Proof of Corollary~\ref{cor:direct}]
Let $\eps > 0$ and $\gamma > 0$, we will fix $\delta > 0$ later.
Assume that the OR of $N$ given instances of \SSUM on target values $T_1, \ldots, T_N = O(N^\gamma)$ and at most
$\delta \log N$ numbers each, can be solved in total time $O(N^{(1 + \gamma)(1 - \eps)})$. We will show that SETH fails.

Let $\phi$ be an instance of \kSAT on $n$ variables. Split the set of variables into $X_1$ and $X_2$ of size $n_1$ and $n_2$, such that $n_2 = \gamma \cdot n_1$ up to rounding. Specifically, we can set $n_1 := \lceil \frac n{1+\gamma} \rceil$ and $n_2 := \lfloor \frac {\gamma n}{1+\gamma} \rfloor$ and thus have $n_2 \le \gamma n_1$. Enumerate all assignments of the variables in $X_1$. For each such assignment $\alpha$ let $\phi_\alpha$ be the resulting \kSAT instance after applying the partial assignment~$\alpha$.

For each $\phi_\alpha$, run the reduction from Theorem~\ref{thm:redsatsubsetsum} with $\eps' = \min\{1/2, 1/\gamma\} \cdot \eps/2$, resulting in at most $2^{\eps' n_2}$ instances of \SSUM on at most $\tilde c_{k,\eps'} n_2$ items and target at most $2^{(1+2\eps') n_2}$. In total, we obtain at most $2^{n_1 + \eps' n_2}$ instances of \SSUM, and $\phi$ is satisfiable iff at least one of these \SSUM instances is a YES-instance. Set $N := 2^{(1+\eps/2)n_1}$ and note that the number of instances is at most $2^{n_1 + \eps' n_2} \le 2^{(1+\gamma \eps')n_1} \le N$, and that the target bound is at most $2^{(1+2\eps') n_2} \le 2^{(1+2 \eps')\gamma n_1} \le N^\gamma$.
Thus, we constructed at most $N$ instances of \SSUM on target at most $N^\gamma$, each having at most $\tilde c_{k,\eps'} n_2 \le \tilde c_{k,\eps'} n$ items.

Using the assumed algorithm, the OR of these instances can be solved in total time $O(N^{(1 + \gamma)(1 - \eps)})$. Since $(1+\gamma) n_1 = (1+\gamma) \lceil \frac n{1+\gamma} \rceil \le n + 1+ \gamma = n + O(1)$ and $(1+\eps/2)(1-\eps) \le 1-\eps/2$, this running time is
\begin{align*}
  O\big(N^{(1 + \gamma)(1 - \eps)}\big) &= O\big(\big(2^{(1+\eps/2)n_1}\big)^{(1+\gamma)(1-\eps)}\big) \\
  &= O\big(2^{(1-\eps/2) n}\big),
\end{align*} 
which contradicts SETH. Specifically, assuming SETH for some $k = k(\eps)$ this running time is less than the time required for \kSAT. Setting $\delta := \tilde c_{k,\eps'}$ finishes the proof.
\end{proof}

\section{The Bicriteria \boldmath{$s,t$}-Path Problem}
\label{sec: sethbicriteria}
% !TEX root = paper.tex

In this section we apply the results of the previous section to the \bipath problem. We will show that the \bipath problem is in fact harder than \SSUM, by proving that the classical pseudo-polynomial time algorithm for the problem cannot be improved on sparse graphs assuming SETH. We also prove Theorem~\ref{thm: lambdachi} concerning a bounded number of different edge-lengths $\lambda$ and edge-costs $\chi$ in the input network, and Theorem~\ref{thm: k} concerning a bounded number $k$ of internal vertices in a solution path.

\subsection{Sparse networks}

We begin with the case of sparse networks; \emph{i.e.} input graphs on $n$ vertices and $O(n)$ edges. We embed multiple instances of \SSUM into one instance of \bipath to prove Theorem~\ref{thm:cspath}, namely that there is no algorithm for \bipath on sparse graphs faster than the well-known $O(\min\{nL,nC\})$-time algorithm.

\begin{proof}[Proof of Theorem~\ref{thm:cspath}]
We show that for any $\eps > 0, \,\gamma > 0$, an algorithm solving \bipath on sparse $n$-vertex graphs and budgets $L,C = \Theta(n^\gamma)$ in time $O(n^{(1+\gamma)(1-\eps)})$ contradicts SETH. As in Corollary~\ref{cor:direct}, let $(Z_1,T_1),\ldots,(Z_N,T_N)$ be instances of \SSUM on targets $T_i \le N^\gamma$ and number of items $|Z_i| \le \delta \log N$ for all $i$. Without loss of generality, we can assume that all sets $Z_i$ have the same size $k = \delta \log N$ (\emph{e.g.}, by making $Z_i$ a multiset containing the number 0 multiple times).

Fix an instance $(Z_i,T_i)$ and let $Z_i = \{z_1, \ldots, z_k\}$. We construct a graph $G_i$ whose vertex set is $\{s, v_1,\ldots, v_k, t\}$. Writing $v_0 := s$ for simplicity, for each $j \in [k]$ we add an edge from $v_{j-1}$ to $v_j$ with length $z_j$ and cost $N^\gamma - z_j$, and we add another\footnote{Note that parallel edges can be avoided by subdividing all constructed edges.} edge from $v_{j-1}$ to $v_j$ with length 0 and cost $N^\gamma$. Finally, we add an edge from $v_k$ to $t$ with length $N^\gamma - T_i$ and cost $T_i$. Then the set of $s,t$-paths corresponds to the power set of $Z_i$, and the $s,t$-path corresponding to $Y \subseteq Z_i$ has total length $N^\gamma - T_i + \sum_{y \in Y} y$ and cost $k N^\gamma + T_i - \sum_{y \in Y} y$. Hence, setting the upper bound on the length to $L = N^\gamma$ and on the cost to $C = k N^\gamma$, there is an $s,t$-path respecting these bounds iff there is a subset $Y$ of $Z_i$ summing to $T_i$, i.e., iff $(Z_i,T_i)$ is a YES-instance.

We combine the graphs $G_1, \ldots, G_N$ into one graph $G$ by identifying all source vertices $s$, identifying all target vertices $t$, and then taking the disjoint union of the remainder. With the common length bound $L = N^\gamma$ and cost bound $C = k N^\gamma$, there is an $s,t$-path respecting these bounds in $G$ iff some instance $(Z_i,T_i)$ is a YES-instance. Furthermore, note that $G$ has $n = \Theta(N \log N)$ vertices, is sparse, and can be constructed in time $O(N \log N)$. Hence, an $O(n^{(1+\gamma)(1-\eps)})$ time algorithm for \bipath would imply an $O(N^{(1+\gamma)(1-\eps)} \textup{polylog} N) = O(N^{(1+\gamma)(1-\eps/2)})$ time algorithm for deciding whether at least one of $N$ \SSUM instances is a YES-instance, a contradiction to SETH by Corollary~\ref{cor:direct}.

Finally, let us ensure that $L, C = \Theta(n^\gamma)$. Note that the budgets $L$ and $C$ are both bounded by $O(N^\gamma \log N)$. If $\gamma \ge 1$, then add a supersource $s'$ and one edge from $s'$ to $s$ with length and cost equal to $N^\gamma \log^\gamma N$, and add $N^\gamma \log^\gamma N$ to $L$ and $C$. This results in an equivalent instance, and the new bounds $L,C$ are $\Theta(N^\gamma \log^\gamma N) = \Theta(n^\gamma)$. If $\gamma < 1$, then do the same where the length and cost from $s'$ to $s$ is $N^\gamma \log N$, and then add $N \log^{1/\gamma} N$ dummy vertices to the graph to increase $n$ to $\Theta(N \log^{1/\gamma} N)$. Again we obtain budgets $L,C = \Theta(N^\gamma \log N) = \Theta((N \log^{1/\gamma} N)^\gamma) = \Theta(n^\gamma)$. In both cases, the same running time analysis as in the last paragraph goes through. This completes the proof of Theorem~\ref{thm:cspath}.
\end{proof}
%
%
%\begin{theorem} \label{thm:cspath}
%  Assuming SETH, for any $\eps > 0$ and $q > 0$ no algorithm solves \bipath on sparse $n$-node graphs and budget $B = \Theta(n^q)$ in time $O((nB)^{1-\eps})$.
%\end{theorem} 

\subsection{Few different edge-lengths or edge-costs}
\label{sec: lambda+chi}
%\input{lambda+chi}

%%%%%%%%%%%%%%%%%%%%%%%%%%%%%%%%%%%%%%%%%%%%%%%%%%%%%%%%%%%%%%%%%%%%%%%%%%%%%%%%%%%%%%%%%%%%%%%%%%%%%%%%%%%%%%%%%%%%%%%
%%%%%%%%%% Networks with a small number of edge lengths or edge costs
%%%%%%%%%%%%%%%%%%%%%%%%%%%%%%%%%%%%%%%%%%%%%%%%%%%%%%%%%%%%%%%%%%%%%%%%%%%%%%%%%%%%%%%%%%%%%%%%%%%%%%%%%%%%%%%%%%%%%%%

We next consider the parameters $\lambda$ (the number of different edge-lengths) and $\chi$ (the number of different edge-costs). We show that \bipath can be solved in $O(n^{\min\{\lambda,\chi\}+2})$ time, while its unlikely to be solvable in $O(n^{\min\{\lambda,\chi\} - 1 -\eps})$ for any $\eps > 0$, providing a complete proof for Theorem~\ref{thm: lambdachi}. The upper bound of this theorem is quite easy, and is given in the following lemma.

\begin{lemma}
\label{lemma: simple xp}
\bipath can be solved in $O(n^{\min\{\lambda,\chi\}+2})$ time.
\end{lemma}

\begin{proof}
It suffices to give an $O(n^{\lambda + 2})$ time algorithm, as the case of time $O(n^{\chi+2})$ is symmetric, and a combination of these two algorithms yields the claim.
Let $\tilde{\ell}_1,\ldots,\tilde{\ell}_\lambda$ be all different edge-length values. We compute a table $T[v,i_1,\ldots,i_\lambda]$, where $v \in V(G)$ and $i_1,\ldots,i_\lambda \in \{0,\ldots,n\}$, which stores the minimum cost of any $s,v$-path that has exactly $i_j$ edges of length $\tilde{\ell}_j$, for each $j \in \{1,\ldots,\lambda\}$. For the base case of our computation, we set $T[s,0,\ldots,0]=0$ and $T[s,i_1,\ldots,i_\lambda]=\infty$ for entries with some $i_j \neq 0$. The remaining entries are computed via the following recursion:
\begin{align*}
T[v,i_1,\ldots,i_\lambda] = 
&\min_{1 \le j \le \lambda} \min_{\substack{(u,v) \in E(G),\\ \ell((u,v)) = \tilde{\ell}_j.}} T[u,i_1,\ldots,i_j-1,\ldots,i_\lambda] + c((u,v)).
\end{align*}

It is easy to see that the above recursion is correct, since if $e_1,\ldots,e_k$ is an optimal $s,v$-path corresponding to an entry $T[v,i_1,\ldots,i_\lambda]$ in $T$, with $e_k=(u,v)$ and $\ell(e_k) = \tilde{\ell}_j$ for some $j \in \{1,\ldots,\lambda\}$, then $e_1,\ldots,e_{k-1}$ is an optimal $s,u$-path corresponding to the entry $T[u,i_1,\ldots,i_j-1,\ldots,i_\lambda]$. Thus, after computing table $T$, we can determine whether there is a feasible $s,t$-path in~$G$ by checking whether there is an entry $T[t,i_1,\ldots,i_\lambda]$ with $\sum_{j=1}^\lambda i_j \cdot \tilde{\ell}_j \leq L$ and $T[t,i_1,\ldots,i_\lambda] \leq C$. As there are $O(n^{\lambda+1})$ entries in $T$ in total, and each entry can be computed in $O(n)$ time, the entire algorithm requires $O(n^{\lambda+2})$ time.
\end{proof}

We now turn to proving the lower-bound given in Theorem~\ref{thm: lambdachi}. The starting point is our lower bound for \kSUM ruling out $O(T^{1-\eps} n^{\delta k})$ algorithms (Theorem~\ref{thm:sethsubsetsum}). We present a reduction from \kSUM to \bipath, where the resulting graph in the \bipath instance has few different edge-lengths and edge-costs.

Let $(Z_1,\ldots,Z_k,T)$ be an instance of \kSUM with $Z_i \subset [0,T]$ and $|Z_i| \le n$ for all $i$, and we want to decide whether there are $z_1 \in Z_1, \ldots, z_k \in Z_k$ with $z_1 + \ldots + z_k = T$. We begin by constructing an acyclic multigraph $G^*$, using similar ideas to those used for proving Theorem~\ref{thm:cspath}. The multigraph $G^*$ has $k+1$ vertices $s=v_0,\ldots,v_k=t$, and is constructed as follows: For each $i \in \{1,\ldots,k\}$, we add at most $n$ edges from $v_{i-1}$ to $v_i$, one for each element in $Z_i$. The length of an edge $e \in E(G^*)$ corresponding to element $z_i \in Z_i$ is set to $\ell(e) = z_i$, and its cost is set to $c(e) = T-z_i$.

\begin{lemma}
\label{lemma: small}%
$(Z_1,\ldots,Z_k,T)$ has a solution iff $G^*$ has an $s,t$-path of length at most $L=T$ and cost at most $C=T (k-1)$.
\end{lemma}

\begin{proof}
Suppose there are $z_1 \in Z_1, \ldots, z_k \in Z_k$ that sum to $T$. Consider the $s,t$-path $e_1,\ldots,e_k$ in $G^*$, where $e_i$ is the edge from $v_{i-1}$ to $v_i$ corresponding to $z_i$. Then $\sum_{i=1}^k \ell(e_i) = \sum_{i=1}^k z_i = T  = L$, and $\sum_{i=1}^k c(e_i) = \sum_{i=1}^k T-z_i = kT-T = C$. Conversely, any $s,t$-path in $G^*$ has $k$ edges $e_1,\ldots,e_k$, where $e_i$ is an edge from $v_{i-1}$ to $v_i$. If such a path is feasible, meaning that $\sum_{i=1}^k \ell(e_i) \leq L=T$ and $\sum_{i=1}^k c(e_i) \leq C=T(k-1)$, then these two inequalities must be tight because $c(e_i) = T-\ell(e_i)$ for each $i \in [k]$. This implies that the integers $z_1,\ldots,z_k$ corresponding to the edges $e_1,\ldots,e_k$ of~$G^*$, sum  to~$T$.
\end{proof}

Let $\tau \ge 1$ be any constant and let $B := \lceil T^{1/\tau} \rceil$. We next convert $G^*$ into a graph $\tilde{G}$ which has $\tau+1$ different edge-lengths and $\tau+1$ different edge-costs, both  taken from the set $\{0,B^0,B^1,\ldots,B^{\tau-1}\}$. Recall that $V(G^*)=\{v_0,\ldots,v_k\}$, and the length and cost of each edge in $G^*$ is non-negative and bounded by $T$. The vertex set of $\tilde{G}$ will include all vertices of $G^*$, as well as additional vertices. %The set of edge-lengths and edge-costs in $\tilde{G}$ will be taken from the set of $\tau$ integers $\{0,n^0,n^1,\ldots,n^{k-1}\}$.

For an edge $e \in E(G^*)$, write its length as $\ell(e) = \sum_{i=0}^{\tau-1} a_i B^i$, and its cost as $c(e) = \sum_{i=0}^{\tau-1} b_i B^i$, for integers $a_1,\ldots,a_{\tau-1},b_1,\ldots,b_{\tau-1} \in \{0,\ldots,B-1\}$. We replace the edge $e$ of $G^*$ with a path in $\tilde{G}$ between the endpoints of $e$ that has $\sum_{i=0}^{\tau-1} (a_i + b_i)$ internal vertices. For each $i \in \{0,\ldots,\tau-1\}$, we set $a_i$ edges in this path to have length $B^i$ and cost 0, and $b_i$ edges to have length 0 and cost $B^i$. Replacing all edges of $G^*$ by paths in this way, we obtain the graph $\tilde{G}$ which has $O(n B)$ vertices and edges (since $k$ and $\tau$ are constant). As any edge in $G^*$ between $v_i$ and $v_{i+1}$ corresponds to a path between these two vertices in $\tilde{G}$ with the same length and cost, we have:
\begin{lemma}
\label{lemma: multi}%
Any $s,t$-path in $G^*$ corresponds to an $s,t$-path in $\tilde{G}$ with same length and cost, and vice-versa.
\end{lemma}

%At the final step of our construction we replace the high degree vertices $v_0,\ldots,v_{k-1}$ of $G$. We replace these vertices with binary out-trees and binary in-trees which have edges with 0 length and cost, except for edges to or from leaves in the tree which retain the original lengths and costs in $G$. It is not difficult to see that if we do this for all high degree vertices in $G$, we obtain a planar directed acyclic graph $G'$ where all in-degrees and out-degrees are bounded by 2, and where each $s,t$-path in
%\begin{lemma}
%\label{lemma: highdeg}%
%Any $s,t$-path in $\tilde{G}$ corresponds to an $s,t$-path in $G$ with same length and cost, and vice-versa.
%\end{lemma}

\begin{lemma}
Assuming SETH, for any constant $\lambda, \chi \ge 2$ there is no $O(n^{\min\{\lambda,\chi\} - 1 -\eps})$ algorithm for \bipath for any $\varepsilon>0$.
\end{lemma}
\begin{proof}
Suppose \bipath has a $O(n^{\min\{\lambda,\chi\} - 1 -\eps})$ time algorithm. We use this algorithm to obtain a fast algorithm for \kSUM, contradicting SETH by Theorem~\ref{thm:sethsubsetsum}.
On a given input $(Z_1,\ldots,Z_k,T)$ of \kSUM on $n$ items, for $\tau := \min\{\lambda,\chi\}-1$ we construct the instance $(\tilde G,s,t,L,C)$ described above. Then $\tilde G$ is a directed acyclic graph with $\tau+1=\min\{\lambda,\chi\}$ different edge-lengths and edge-costs $\{0,B^0,B^1,\ldots,B^{\tau-1}\}$. Moreover, due to Lemmas~\ref{lemma: small} and \ref{lemma: multi}, there are $z_1 \in Z_1, \ldots, z_k \in Z_k$ summing to $T$ iff $\tilde G$ has a feasible $s,t$-path. Thus, we can use our assumed \bipath algorithm on $(\tilde G,s,t,L,C)$ to solve the given \kSUM instance. As $\tilde G$ has $O(n B)$ vertices and edges, where $B = \lceil T^{1/\tau} \rceil$, an $O(n^{\min\{\lambda,\chi\} - 1 -\eps})$ algorithm runs in time $O((n B)^{\tau - \eps}) = O(T^{1-\eps/\tau} n^{\tau})$ time on $(\tilde G,s,t,L,C)$. For $\delta := \delta(\eps/\tau)$ from Theorem~\ref{thm:sethsubsetsum} and $k$ set to $\tau / \delta$, this running time is $O(T^{1-\eps/\tau} n^{\delta(\eps/\tau) k})$ and thus contradicts SETH by Theorem~\ref{thm:sethsubsetsum}.
\end{proof}

\subsection{Solution paths with few vertices}
\label{sec: exact}
\def \kBSP {$k$-BSP }

In this section we investigate the complexity of \bipath with respect to the number of internal vertices $k$ in a solution path. Assuming $k$ is fixed and bounded, we obtain a tight classification of the time complexity for the problem, up to sub-polynomial factors, under Conjecture~\ref{conj:ksum}.

Our starting point is the \EkPath problem: Given an integer $T \in \{0,\ldots,W\}$, and a directed graph $G$ with edge weights, decide whether there is a simple path in $G$ on $k$ vertices in which the sum of the weights is exactly $T$. Thus, this is the "exact" variant of \bipath on graphs with a single edge criterion, and no source and target vertices. The \EkPath problem can be solved in $\tilde{O}(n^{\lceil(k+1)/2\rceil})$ time by a ``meet-in-the-middle" algorithm~\cite{AL13}, where the $\tilde{O}(\cdot)$ notation suppresses poly-logarithmic factors in $W$. It is also known that \EkPath has no $\tilde{O}(n^{\lceil(k+1)/2\rceil-\eps})$ time algorithm, for any $\eps > 0$, unless the \textsc{$k$-Sum} conjecture is false~\cite{AL13}. We will show how to obtain  similar bounds for \bipath by implementing a very efficient reduction between the two problems.

To show that \EkPath can be used to solve \bipath, we will combine multiple ideas. The first is the observation that \EkPath can easily solve the \EkBiPath problem, a variant which involves bicriteria edge weights: Given a pair of integers $(T_1,T_2)$, and a directed graph $G$ with two edge weight functions $w_1(\cdot)$ and $w_2(\cdot)$, decide whether there is a simple path in $G$ on $k$ vertices in which the sum of the $w_i$-weights is exactly $T_i$ for $i \in \{1,2\}$.

\begin{lemma}
\label{twod}
There is an $O(n^2)$ time reduction that reduces an instance of \EkBiPath with edge weights in $\{0,1,\ldots,W\}^2$ to an instance of \EkPath with edge weights in $\{0,1,\ldots, 2kW^2 + W\}$.
\end{lemma}

\begin{proof}
Define a mapping of a pairs in $\{0,1,\ldots,W\}^2$ to single integers $\{0,\ldots,2kW^2 + W\}$ by setting $f(w_1,w_2)=w_2+w_1\cdot 2kW$ for each $w_1,w_2 \in\{0,\ldots,W\}$. Observe that for any~$k$ pairs $(w^1_1,w^1_2),\ldots,(w^k_1,w^k_2)$, we have $(\sum_{i=1}^k w^i_1 = T_1 \wedge \sum_{i=1}^k w^i_2 = T_2)$ iff $\sum_{i=1}^k f(w^i_1,w^i_2) = f(T_1,T_2)$. Therefore, given a graph as in the statement, we can map each pair of edge weights into a single edge weight, thus reducing to \EkPath without changing the answer.
\end{proof}

The next and more difficult step is to reduce \bipath to \EkBiPath. This requires us to reduce the question of whether there is a path of length and cost at most $L$ and $C$, to questions about the existence of paths with length and cost \emph{equalling exactly} $T_1$ and $T_2$. A naive approach would be to check if there is a path of exact length and cost $(T_1,T_2)$ for all values $T_1 \leq L$ and $T_2 \leq C$. Such a reduction will incur a very large $O(LC)$ overhead. We will improve this to $O(\log^{O(1)}{(L+C)})$.

In the remainder of this section, let $W$ be the maximum of $L$ and $C$. The idea behind our reduction is to look for the smallest $x,y \in [\log{W}]$ such that if we restrict all edge lengths $\ell$ to the $x$ most significant bits of $\ell$, and all edge costs $c$ to the $y$ most significant bits of $c$, then there is a path that satisfies the threshold constraints with equality. To do this, we can check for every pair $x,y$, whether after restricting edge lengths and costs, there is a $k$-path with total weight \emph{exactly} equal to the restriction of the vector $(L,C)$, possibly minus the carry from the removed bits. Since the carry from summing $k$ numbers can be at most $k$, and we have to consider this carry for the length and the cost, we will not have to check more than $O(k^2)$ ``targets" per pair $x,y \in [\log{W}]$.

To implement this formally, we will need the following technical lemma. The proof uses a bit scaling technique that is common in approximation algorithms. Previously, tight reductions that use this technique were presented by Vassilevska and Williams~\cite{VW09} (in a very specific setting), and by Nederlof \emph{et al.}~\cite{NLZ12} (who proved a general statement). We will need a generalization of the result of \cite{NLZ12} in which we introduce a parameter $k$, and show that the overhead depends only on $k$ and $W$, and does not depend on~$n$.

\begin{lemma}
\label{mintoexact}%
Let $U$ be a universe of size $n$ with weight functions $w_1,w_2: U \to \{0,\ldots,W\}$, and let $T_1,T_2 \in \{0,\ldots,W\}$ be integers.
Then there is a polynomial time algorithm that returns a set of weight functions $w_1^{(i)},w_2^{(i)}: U \to \{0,\ldots,W\}$ and integers $T_1^{(i)},T_2^{(i)} \in \{0,\ldots, W\}$, for $i \in [q]$ and $q=O(k^2 \log^2{W})$, such that: For every subset $X \subseteq U$ of size $|X|=k$ we have $(w_1(X) \leq T_1  \wedge w_2(X) \leq T_2)$ if and only if there exists an $i \in [q]$ with $(w^{(i)}_1(X) = T^{(i)}_1  \wedge w^{(i)}_2(X) = T^{(i)}_2)$.
\end{lemma}

\begin{proof}
Note that for any subset $X \subseteq U$ with $|X| = k$ we have $w_1(X), w_2(X) \in \{0,\ldots,kW\}$. 
We will assume that a number in $ \{0,\ldots,kW\}$ is encoded in binary with $\log(kW)$ bits in the standard way. For numbers $a \in \{0,\ldots,kW\}$ and $x \in [\log{W}]$ we let $[a]_x = \lfloor a / 2^{x} \rfloor$, that is, we remove the $x$ least significant bits of $a$. In what follows, we will construct weight functions and targets for each dimension independently and in a similar way. We will present the construction for the $w_1$'s.

First, we add the weight functions $w_1^{(i)}=w_1$, with target $T^{(i)}_1=T_1-a$ for any $a \in [4k]$. Call these the initial $(i)$'s. Then, for any $x \in [\log{W}]$ and $a \in [2k]$, we add the weight function $w_1^{(i)} (e)=[w_1(e)]_x$, and set the target to $T_1^{(i)}=[T_1]_x - k - a$. This defines $O(k \log{W})$ new functions and targets, and we will show below that for any subset $X \subseteq U$ we have that $w_1(X) \leq T_1$ iff for some $i$ we have $w_1^{(i)}(X) = T_1^{(i)}$. Then, we apply the same construction for $w_2$, and take every pair of constructed functions and targets, to obtain a set of $O(k^2 \log^2{W})$ functions and targets that satisfy the required property.

The correctness will be based on the following bound, which follows because when summing $k$ numbers the carry from removed least significant bits cannot be more than $k$. In particular, for any $x \in [\log W]$ and $a_1,\ldots,a_k \in \{0,\ldots,W\}$, we have
$$
\sum_{j=1}^k [a_j ]_x \leq \left[\sum_{j=1}^k a_j \right]_x \leq k + \sum_{j=1}^k [a_j ]_x.
$$

Fix some $X \subset U$. For the first direction, assume that for some $i$, $w_1^{(i)}(X) = T^{(i)}_1$.
If it is one of the initial $(i)$'s, then we immediately have   $w_1(X)\leq T_1$.
Otherwise, if $X = \{ v_1,\ldots, v_k\}$ then %we get that
\begin{align*}
[w_1(X)]_x = \left[ \sum_{j=1}^k w_1(v_j) \right]_x &\leq k + \sum_{j=1}^k [w_1(v_j) ]_x \\
&= w_1^{(i)}(X) +k \\
&= T^{(i)}_1 +k \leq [T_1]_x -1,
\end{align*}
which implies that $w_1(X) < T_1$.

For the other direction, assume that $w_1(X) \le T_1$. If $w_1(X) \geq T_1-4k$ then for one of the initial $(i)$'s we will have $w_1^{(i)}(X) = w_1(X) = T_1-a = T_1^{(i)}$ for some $a \in [4k]$. Otherwise, let $x$ be the largest integer in $[\log W]$ for which $[w_1(X)]_x \leq [T_1]_x - k$.
Because $x$ is the largest, we also know that  $[w_1(X)]_x \geq [T_1]_x - 2k$. Therefore,
$$
w_1^{(i)}(X) = \sum_{j=1}^k [w_1(v_j) ]_x \leq \left[\sum_{j=1}^k w_1(v_j) \right]_x \leq [T_1]_x -k,
$$
and
$$
w_1^{(i)}(X) = \sum_{j=1}^k [w_1(v_j) ]_x \geq \left[\sum_{j=1}^k w_1(v_j) \right]_x -k \geq [T_1]_x -3k.
$$
It follows that for some $a \in [2k]$, we have $w_1^{(i)}(X) = [T_1]_x-k-a = T_1^{(i)}$.
\end{proof}

We are now ready to present the main reduction of this section. Let $(G,s,t,L,C)$ be a given instance of \bipath. Our reduction follows three general steps that proceed as follows:
\begin{enumerate}
\item \emph{Color coding:} At the first step, we use the derandomized version of the color coding technique~\cite{AlonYZ95} to obtain $p'=2^{O(k)}\log n=O_k(\log n)$ partitions of the vertex set $V(G) \setminus \{s,t\}$ into $k$ classes $V^{(\alpha)}_1,\ldots,V^{(\alpha)}_k$, $\alpha \in [p']$, with the following property: If there is a feasible $s,t$-path $P$ with $k$ internal vertices in $G$, we are guaranteed that for at least one partition we will have $|V(P) \cap V^{(\alpha)}_i| = 1$ for each $i \in [k]$. By trying out all possible $O_k(1)$ orderings of the classes in each partition, we can assume that if $P=s,v_1,\ldots,v_k,t$, then $V(P) \cap V^{(\alpha)}_i = \{v_i\}$ for each $i \in [k]$.

    \hspace{15pt} Let $p$ denote the total number of ordered partitions. For each ordered partition $\alpha \in [p]$, we remove all edges between vertices inside the same class, and all edges $(u,v)$ where $u \in V^{(\alpha)}_i$, $v \in V^{(\alpha)}_j$, and $j \neq i+1$. We also remove all edges from $s$ to vertices not in $V^{(\alpha)}_1$, and all edges to $t$ from vertices not in $V_k$. Let~$G_\alpha$ denote the resulting graph, with $\alpha \in [p]$ for $p=O_k(\log n)$.\\

\item \emph{Removal of $s$ and $t$:} Next, we next remove $s$ and $t$ from each $G_\alpha$. For every vertex $v \in V_1^{(\alpha)}$, if $v$ was connected with an edge from $s$ of length $\ell$ and cost $c$, then we remove this edge and add this length $\ell$ and cost $c$ to all the edges outgoing from $v$. Similarly, we remove the edge from $v \in V^{(\alpha)}_k$ to $t$, and add its length and cost to all edges ingoing to $v$. Finally, any vertex in $V^{(\alpha)}_1$ that was not connected with an edge from $s$ is removed from the graph, and every vertex in $V^{(\alpha)}_k$ that was not connected to $t$ is removed.\\
\item \emph{Inequality to equality reduction:} Now, for each $G_\alpha$, we apply Lemma~\ref{mintoexact} with the universe $U$ being the edges of $G_\alpha$, and $w_1,w_2:U \to \{0,\ldots,W\}$ being the lengths and costs of the edges. We get a set of $q=O_k(\log^2{W})$ weight functions and targets. For $\beta \in [q]$, let $G_{\alpha,\beta}$ be the graph obtained from $G$ by replacing the lengths and costs with new functions $w^{(\beta)}_1,w^{(\beta)}_2$. The final \EkBiPath is then constructed as $(G_{\alpha,\beta},T^{(\beta)}_1,T^{(\beta)}_2)$.
\end{enumerate}

Thus, we reduce our \bipath instance to at most $O_k(\log{n} \log^2{W})$ instances of \EkBiPath. Note that if $G$ contains a feasible $s,t$-path $P=s,v_1,\ldots,v_k,t$ of length $\ell_P \leq L$ and cost $c_P$, then by correctness of the color coding technique, there is some $\alpha \in [p]$ such that $G_\alpha$ contains $P$ with $V(P) \cap V^{(\alpha)}_i = \{v_i\}$ for each $i \in [k]$. Moreover, the total weight of $v_1,\ldots,v_k$ in $G_\alpha$ is $(\ell_P,c_P)$. By Lemma~\ref{mintoexact}, there is some $\beta \in [q]$ for which the total weight of $v_1,\ldots,v_k$ in $G_{\alpha,\beta}$ is $(T^{(\beta)}_1,T^{(\beta)}_2)$. Thus, $P$ is a solution for $(G_{\alpha,\beta},T^{(\beta)}_1,T^{(\beta)}_2)$. Conversely, by  the same line of arguments, any solution path for some \EkBiPath instance $(G_{\alpha,\beta},T^{(\beta)}_1,T^{(\beta)}_2)$ corresponds to a feasible $s,t$-path in $G$ with $k$ internal vertices.

Thus, we have obtained a reduction from \bipath to \EkBiPath. Combining this with reduction from \EkBiPath to \EkPath given in Lemma~\ref{twod}, we obtain the following.

\begin{lemma}
\label{final}
Fix $k \ge 1$, and let $(G,s,t,L,C)$ be an instance of \bipath where $G$ has $n$ vertices. Set $W=\max \{L,C\}$. Then one can determine whether $(G,s,t,L,C)$ has a solution with $k$ internal vertices by solving $O_k(\log{n} \log^2{W})$ instances of \EkPath on graphs with $O(n)$ vertices and edge weights bounded by $O_k(W^2)$.
\end{lemma}

\begin{corollary}
\label{cor:exact}
For any fixed $k \geq 1$ there is an algorithm solving \bipath on $k$ vertices in $\tilde{O}(n^{\lceil(k+1)/2\rceil})$ time.
\end{corollary}

\begin{proof}
By Lemma~\ref{final}, an instance of \bipath can be reduced to $O(\log{n} \log^2{W})$ instances of \EkPath.
Using the algorithm in \cite{AL13}, each of these \EkPath instances can be solved in $\tilde{O}(n^{\lceil(k+1)/2\rceil})$ time.
\end{proof}

We next turn to proving our lower bound for \bipath. For this, we show a reduction in the other direction, from \EkPath to \bipath.

\begin{lemma}
\label{lemma: lower bound k}
Let $\eps > 0$ and $k \geq 1$. There is no $\tilde{O}(n^{\lceil(k+1)/2\rceil - \eps})$ time algorithm for \bipath on $k$ vertices unless the \textsc{$k$-Sum} conjecture (Conjecture~\ref{conj:ksum}) is false.
\end{lemma}

\begin{proof}
We show a reduction from \EkPath to \bipath. This proves the claim, as it is known that an  $\tilde{O}(n^{\lceil(k+1)/2\rceil-\eps})$ time algorithm for \EkPath, for any $\eps > 0$, implies that the \textsc{$k$-Sum} conjecture is false~\cite{AL13}. Let $(G,T)$ be an instance of \EkPath, where $G$ is an edge-weighted graph and $T\in\{0,\ldots,W\}$ is the target. We proceed as follows: As in the upper-bound reduction, we first apply the color-coding technique~\cite{AlonYZ95} to obtain $p=O(\log{n})$ vertex-partitioned graphs $G_1,\ldots,G_p$, where $V(G_\alpha)$ is the disjoint union $V^{(\alpha)}_1 \uplus \cdots \uplus V^{(\alpha)}_k$ for each $\alpha \in [p]$, such that~$G$ has a solution path $P=v_1,\ldots,v_k$ iff for at least one graph $G_\alpha$ we have $V(P)=V^{(\alpha)}_i \cap \{v_i\}$ for each~$i \in [k]$.

We then construct a new graph $H_\alpha$ from each graph $G_\alpha$ as follows: We first remove from $G_\alpha$ all edges inside the same vertex class $V^{(\alpha)}_i$, and all edges between vertices in $V^{(\alpha)}_i$ and vertices in~$V^{(\alpha)}_j$ with $j \neq i+1$. We then replace each remaining edge with weight $x \in \{0,\ldots,W\}$ in $G_\alpha$ with an edge with length $x$ and cost $W-x$ in $H_\alpha$. Then, we add vertices $s,t$ to $H_\alpha$, connect $s$ to all the vertices in $V^{(\alpha)}_1$, connect all the vertices in $V^{(\alpha)}_k$ to $t$, and set the length and cost of all these edges to $0$. To complete the proof, we argue that $G$ has a simple path of weight exactly $T$ iff some $H_\alpha$ contains a feasible $s,t$-path for $L= T$ and $C= (k-1)W-T$.

Suppose $P=v_1,\ldots,v_k$ is a simple path in $G$ with $w(P)=T$. Then there is some $\alpha \in [p]$ such that $P$ is a path in $G_\alpha$ with $V(P)=V^{(\alpha)}_i \cap \{v_i\}$ for each~$i \in [k]$. By construction of $H_\alpha$, $P'=s,v_1,\ldots,v_k,t$ is a path in $H_\alpha$, and it has total length $\ell(P')=w(P)=T \leq L$, and total cost $c(P') = (k-1)W-w(P)= (k-1)W-T \leq C$. Conversely, if $P'=s,v_1,\ldots,v_k,t$ is a feasible $s,t$-path in some $H_\alpha$ with length $\ell(P') \leq L$ and cost $c(P') \leq C$, then $P=v_1,\ldots,v_k$ is path in $G$. We know that the weight of $P$ in $G$ is bounded by above by $w(P) = \ell(P') \leq L = T$. Furthermore, we have $(k-1)W-w(P) = (k-1)W-\ell(P') = c(P') \leq C = (k-1)W-T$, implying that $w(P) \geq T$. These two inequalities imply $w(P)=T$, and thus $P$ is a solution for $(G,T)$.

Thus, we can solve $(G,T)$ by solving $O(\log n)$ instances of \bipath. This means that an $\tilde{O}(n^{\lceil(k+1)/2\rceil-\eps})$ algorithm for \bipath, for $\eps > 0$, would imply an algorithm with the same running time for \EkPath. By the reductions in~\cite{AL13}, this refutes the \textsc{$k$-Sum} conjecture.
\end{proof}

Theorem~\ref{thm: k} now immediately follows from the upper and lower bounds given in Corollary~\ref{cor:exact} and Lemma~\ref{lemma: lower bound k} for finding a solution for a \bipath instance that has $k$ internal vertices.

%\subsection{Parameterized complexity}
%\label{sec: fpt}
%\input{fpt}

\medskip
\subsection*{Acknowledgements}
We would like to thank Jesper Nederlof for an inspiring discussion on Subset Sum.

\emph{A.A.} was supported by the grants of Virginia Vassilevska Williams: NSF Grants CCF-1417238, CCF-1528078 and CCF-1514339, and BSF Grant BSF:2012338.
\emph{K.B.:} This work is part of the project TIPEA that has received funding from the European Research Council (ERC) under the European Unions Horizon 2020 research and innovation programme (grant agreement No.\ 850979).
\emph{D.H.} has received funding from the People Programme (Marie Curie Actions) of the European Union's Seventh Framework Programme (FP7/2007-2013) under REA grant agreement number 631163.11, and by the ISRAEL SCIENCE FOUNDATION (grant No. 551145/).
%\section{Discussion}
%\label{sec:discussion}
%\input{discussion}

%\bibliographystyle{plain}
%\bibliography{biblo}

\end{document}